\documentclass[11pt]{amsart} 
\usepackage{graphicx}
\usepackage{color}

\graphicspath{{./figs/}}
\usepackage{wrapfig}
\usepackage[all]{xy}

\usepackage{amsmath,amssymb,amsbsy,amsfonts,mathrsfs}
\usepackage{euscript,stmaryrd}
\usepackage{enumerate}
\usepackage{xcolor}

\newcommand{\N}{{\mathbb{N}}}

\newcommand{\A}{{\tt A}}
\newcommand{\B}{{\tt B}}

\newcommand{\trans}[1]{\stackrel{#1}{\longrightarrow}}
\newcommand{\Amin}{A_{\rm{min}}}

\newtheorem{theorem}{Theorem}
\newtheorem{prop}{Proposition}
\newtheorem{corollary}{Corollary}[theorem]
\newtheorem{lemma}{Lemma}
\newtheorem{conj}{Conjecture}

\theoremstyle{definition}
\newtheorem{definition}{Definition}

\theoremstyle{remark}
\newtheorem{example}{Example}
\newtheorem{remark}{Remark}

\begin{document}

\title{Decidability of
regular language genus computation}
\author{Guillaume Bonfante, Florian Deloup}

\date{}

\begin{abstract}
The article continues the study of the genus of regular languages that the authors introduced in a 2012 paper. Generalizing a previous result, we produce a new family of regular languages on a two-letter alphabet having arbitrary high genus.
Let $L$ be a regular language. In order to understand the genus $g(L)$ of $L$, we introduce the topological size of $|L|_{\rm{top}}$ to be the minimal size of all finite deterministic automata of genus $g(L)$ computing $L$.
We show that the minimal finite deterministic automaton of a regular language can be arbitrary far away from a finite deterministic automaton realizing the minimal genus and computing the same language, both in terms of the difference of genera and in  terms of the difference in size. In particular, we show that the topological size $|L|_{\rm{top}}$ can grow at least exponentially in size $|L|$. We conjecture however the genus of every regular language to be computable. This conjecture implies in particular that the planarity of a regular language is decidable, a question asked in 1976 by R.V. Book and A.K. Chandra. We prove here the conjecture for a fairly generic class of regular languages having no short cycles. 
\end{abstract}

 \maketitle

 \tableofcontents

\section{Introduction}

Regular languages form a robust and well-studied class of languages: they are recognized by finite
deterministic automata (DFA), as well as various formalisms such as Monadic Second-Order logic, finite monoids, regular expressions. Traditionally, the canonical measure of the complexity of a regular language is given by the number of states of its minimal deterministic automaton.

In this paper, we study an alternative measure of language complexity, with a more topological flavor. We will be interested in the topological complexity of underlying graph structures of deterministic automata recognizing the language. Recall that the genus of an oriented surface $\Sigma$ is the maximum number of mutually disjoint 
simple closed curves $C_1, \ldots, C_g \subset \Sigma$ such that the complement $\Sigma - (C_1 \cup \cdots \cup C_g)$ remains connected. This yields a natural notion of genus of a graph: a graph has genus $n$ if it is embeddable in a surface of genus $n$ but cannot be embedded in a surface of strictly smaller genus.

This definition was used in \cite{BD} to define the genus of a regular language $L$ as the minimal genus among underlying graphs of deterministic automata recognizing $L$. In particular, $L$ has genus $0$ if and only if it can be recognized by a planar deterministic automaton. Here we provide new hierarchies of regular languages based on the genus, including for regular languages on two letters (Theorem \ref{th:hierarchy-2-letters}).

One of the main questions is the computability of the genus of a regular language (Conjecture \ref{conj:computable} below). This conjecture implies the decidability of
the planarity of a regular language -- a question raised in 1976 by R.V. Book and A.K. Chandra \cite{BC}. In this paper, we prove the conjecture for the class of regular languages having no short cycle (Theorem \ref{th:computable-class}).

The complexity of the computation of the genus is reflected on the 
cost of extra states needed to build 
a deterministic automaton of minimal genus. We show that the number of states required may be exponential in the size of the minimal automaton of the language (Theorem \ref{th:exponential_size}).

An approach to the computation of the genus of a regular language $L$ consists in considering all possible underlying directed graphs of the automata computating the same language $L$. This leads to the notion of directed emulator of a graph. In several aspects the notion is both similar to and distinct from the classical notion of emulator of a graph (see, for instance, \cite{Hlineny10} for background and a survey on a related open question in graph theory). The main result is that the existence of a directed emulator of genus $g$ of the underlying directed graph of the minimal automaton of a regular language is equivalent to the existence of a deterministic automaton of genus $g$ computing the same language (Theorem \ref{th:genus-and-emulator}). 

\vspace{0.6cm}

\noindent {\textbf{Plan of the paper}}. Section \ref{sec:genus-topsize} provides background, definitions of genus and topological sizes, examples (including hierarchy based on genus, exponential gap between size and topological size) and main results including the computability of the genus of a regular language for a class of regular languages without short cycles. 
Section \ref{sec:diem} provides the set-up of
directed emulators and the main equivalence between finding the genus of a regular language $L$ and finding the minimal genus of directed emulators of the underlying directed graph of the minimal automaton for $L$. The proofs of the main results are collected in Section \ref{sec:proofs}.

\section{The genus and topological size of a regular language} 
\label{sec:genus-topsize}

\subsection{Introductory examples} \label{subsec:examples}

The Myhill-Nerode theorem provides
constructive existence and uniqueness of a deterministic finite automaton with minimal number
of states recognizing a given regular language.

\begin{definition} \label{def:reglang_cyclic_example}
For each $k \geq 1$,
we define the regular language on alphabet $\mathbb{Z}/k \mathbb{Z}$:
$$Z_{k}:=\{ a_1a_2\dots a_n~|~  \sum_{i=1}^n a_{i} \equiv 0 \mod\ k\}.$$
It will be convenient to denote
$Z_{k}^{a_{1}, \ldots, a_{r}}$ the regular language obtained from $Z_{k}$ by restriction to the subalphabet
$\{ a_{1}, \ldots, a_{r} \} \subseteq \mathbb{Z}/k{\mathbb{Z}}$. 

\end{definition}

\begin{example}
The language $Z_{5}^{0,1,2}$. The figure below depicts the minimal automaton $\A$. 
The transitions are of the form $i\trans{j} i+j\mod 5$.

\begin{figure}[!h]\label{fig:Z5012}
\begin{center}
\includegraphics[scale=0.7]{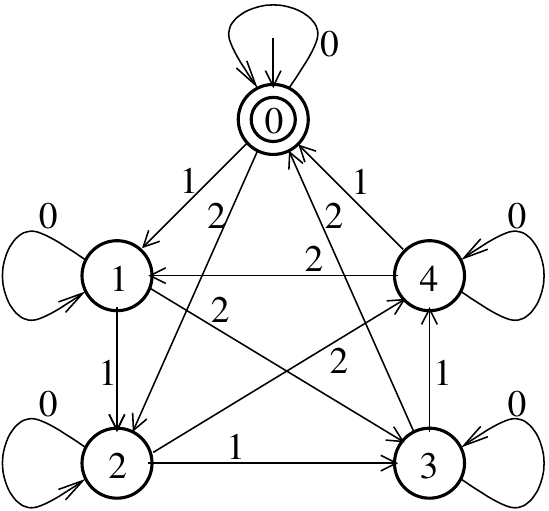}
\end{center}
\caption{The minimal automaton for the language $Z_{5}^{0,1,2}$.}
\end{figure}

Since it contains the complete graph $K_5$, $\A$ is not planar.
However, there exists a deterministic automaton with six states that is planar and computes the same language $L$:

\begin{figure}[!h]\label{fig:plan5}
\begin{center}
\includegraphics[scale=0.6]{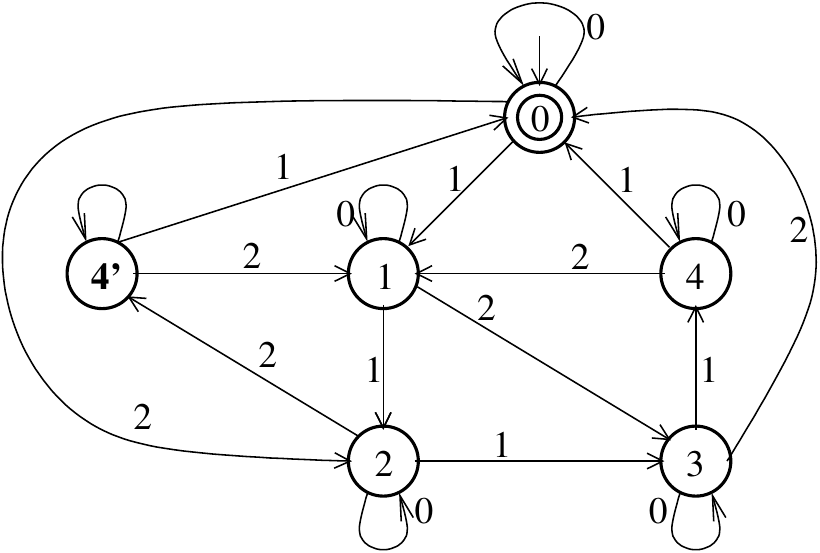}
\end{center}
\caption{A planar automaton $\B$ computing $L$. Note that states $4$ and $4'$ are equivalent.} 
\end{figure}
\end{example}

In the previous example, adding just an extra state suffices to produce a planar equivalent automaton.
The following example suggests that the general case may require many more states.

\begin{example}
The language $Z_{6}$. The figure below represents the minimal
deterministic finite automaton $\A$ 
computing $Z_{6}$. It state space is $\mathbb{Z}/6\mathbb{Z}$ and its transitions are $i\trans{j} i+j\mod 6$, for all $i, j \in \mathbb{Z}/6\mathbb{Z}$.

\begin{figure}[!h]
\begin{center}
\includegraphics[scale=0.7]{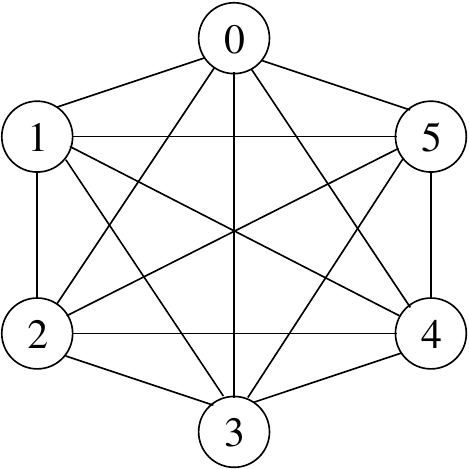}
\end{center}
\caption{The minimal automaton of $Z_{6}$. For simplicity, the self-loop labelled $0$ at each vertex is omitted and each edge represents two transitions in opposite directions.}
\end{figure}

There is
no planar representation for $\A$. (Since $\A$ has the complete graph $K_6$ as a minor, $\A$ is not planar.) However, there exists a deterministic automaton
with $12$ states that is planar and computes the same language
$L$ (Figure below). We regard the additional six states as the price to pay in
order to simplify the topology of an embedding of the automaton
into a surface. Since any $6$-state automaton has an underlying graph which is a subgraph of $Z_6$, it follows easily that any language of size $|L| \leq 6$ can be represented by a planar finite deterministic automaton with at
most $12$ states.

\begin{figure}[!h]
\begin{center}
\includegraphics[scale=0.5]{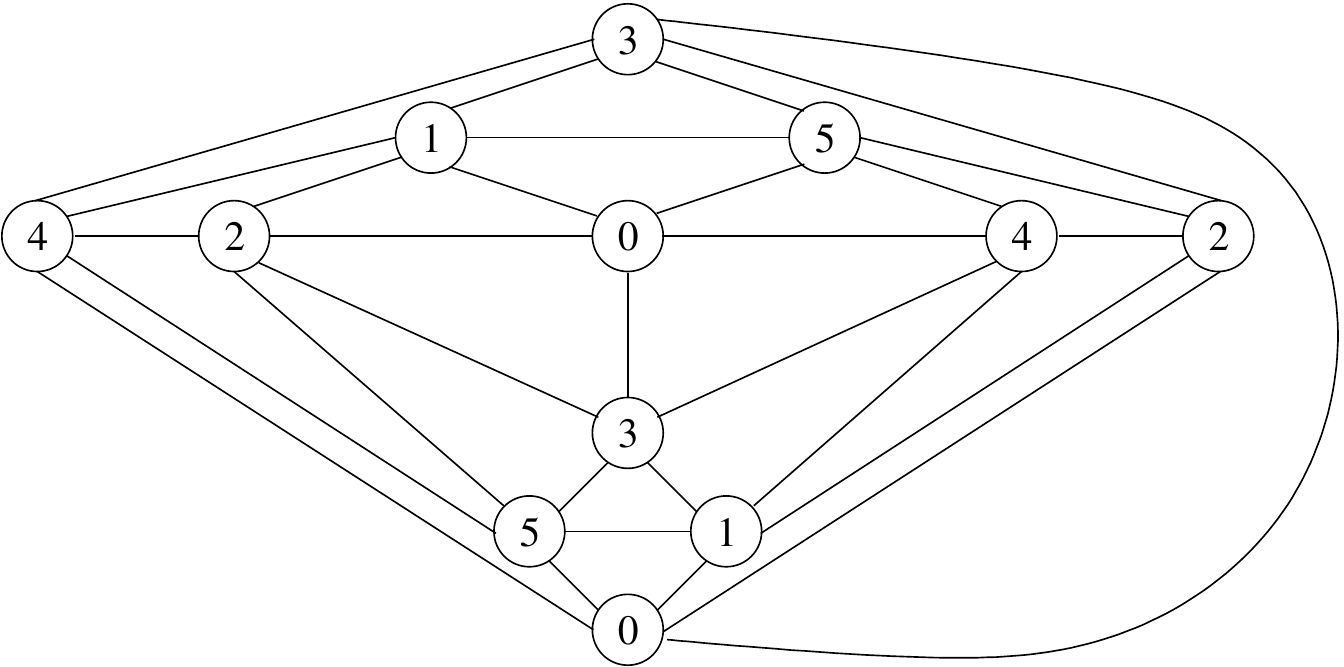}
\end{center}
\caption{A deterministic automaton of minimal genus (planar) recognizing the same
language $Z_6$ (with the same representation conventions as in the previous figure).} \label{fig:LK6}
\end{figure}

\end{example}

%



\subsection{Automata and graphs} \label{sec:automata_and_graphs}

Any automaton $\A$ gives rise to an undirected multigraph (by forgetting labels and orientations of transitions). 
Let $k \geq 1$. A {\emph{cycle of length}} $k$ in $\A$ is a closed walk in the underlying undirected multigraph of length $k$, considered up to circular permutation. Note that a cycle may or may not respect the orientation of the original transitions. We say that the cycle {\emph{respects the direction}} of (the underlying directed multigraph of) $\A$ if each oriented edge of the cycle 
respects the orientation of the original transition in $\A$.
A cycle of length $1$ is also called a {\emph{loop}} (or a {\emph{self-loop}}, for emphasis). A cycle is {\emph{simple}} if it is represented by a closed walk in which no edge is used more than once. In particular, a closed walk in which one edge is travelled twice in opposite directions does not induce a simple cycle. 
The {\emph{genus}} $g(\A)$ of an automaton $\A$ is defined as the genus of the underlying undirected multigraph (see, e.g., \cite[\S 1.4.6]{GT}).


\subsection{Genus-based hierarchies}

We start with the definition of the genus of a regular language \cite{BD}. 

\begin{definition}
Let $L$ be a regular language. Let DFA$(L)$ be the set of all deterministic automata computing $L$. The {\emph{genus}} $g(L)$ is
$$ g(L) = \min \{ g(\A) \ | \ \A \in {\rm{DFA}}(L) \}.$$
A regular language is said to be {\emph{planar}} (resp. {\emph{toric}}) if its genus is zero (resp. one). 
\end{definition}

In other words, the genus of a regular language is the minimal genus 
among all genera of all embeddings of all finite deterministic automata recognizing the language. There are many nonplanar languages. A hierarchy of languages of strictly increasing genus is explicitly constructed in \cite{BD}. 
We produce other examples of such 
a hierarchy (see also Remark \ref{rem:observations}):

\begin{theorem} \label{th:new-hierarchy-example}
Let $k \geq 4$.
The language $Z_{2k+1}^{1,2,\ldots, k}$ has genus $\lceil \frac{(2k-2)(2k-3)}{12} \rceil$. In particular, $g(Z_{2k+1}^{1,2,\ldots, k}) \underset{k \to +\infty}{\to} + \infty$.
\end{theorem}

Note the closed formula for the genus.  In general, the computation of the genus is nontrivial, as shall be explained further below. Note that the family of languages in Theorem \ref{th:new-hierarchy-example} has an increasingly large alphabet. The examples provided in \cite{BD} have a fixed $4$-letter (or more) alphabet. This left out regular languages on an alphabet with fewer letters, namely $2$ or $3$ letters. (Regular languages on a $1$-letter alphabet are easily seen to be planar. See e.g. \cite{BD}.) R.V. Book and A.K. Chandra have built a regular language on $2$ letters that is nonplanar \cite{BC}. We shall prove here the following result.

\begin{theorem} \label{th:hierarchy-2-letters}
There is a genus hierarchy of regular languages on only $2$-letters: for any nonnegative integer $n \geq 0$, there exists a regular language $L$ on a $2$-letter alphabet such that $g(L) = n$.
\end{theorem}

The result is constructive and explicit; it also implies the existence of a genus hierarchy of regular languages on any $k$-letter alphabet for $k \geq 2$ (since self-loops with arbitrary labels can always be added without affecting the genus).

\subsection{Genus and topological size}

Given a regular language $L$, we let $\Amin(L) = \Amin$ be the minimal deterministic automaton associated to $L$. The {\emph{size}} $|L|_{\rm{set}}$ of the language $L$ is the size of the minimal deterministic automaton 
$\Amin$: $$ |L|_{\rm{set}} = |\Amin|.$$ 

\begin{definition}
We define the {\emph{topological size}} of $L$ to be
$$ |L|_{\rm{top}} = \min \{ |A| \ | \ L(A) = L, \ g(A) = g(L)\}$$ where the minimum is taken over all finite deterministic automata recognizing $L$ of minimal genus.
\end{definition}

By definition $|L|_{\rm{top}} \geq |L|_{\rm{set}}$ with equality if and only if the minimal automaton realizes the genus of $L$. From \cite[\S 5]{BD} we know that the topological size is in general reached by several inequivalent deterministic automata. In light of the previous examples, a number of natural questions arise. What is the trade-off between size and genus ? Can a regular language be planar and its minimal automaton
have an arbitrary high genus ? Indeed, 
the following result shows that the topological size of $L$ can grow at least exponentially in terms of (set-theoretic) size of $L$:

\begin{theorem}\label{th:exponential_size}
There is a family of planar regular languages $(L_n)_{n \in \N}$ and a positive number $K>1$ such that $$ |L_{n}|_{\rm{top}} = O(K^{|L_n|_{\rm{set}}}).$$
\end{theorem}

The construction consists in building a sequence of planar languages $L_{n}$ having increasingly high genus minimal automata $\Amin(L_{n})$. 
The language $L_{n}$ will be finite, so there will be a spanning tree for $L_{n}$, ensuring planarity, while the high genus of the minimal automaton is produced by means of a cascade of $n$ directed $K_{5,5}$'s, completed by one initial state and one single final state. See \ref{sec:exponential_size} for details of the proof.

In order to study the genus of a regular language, we introduce some classes of regular languages that ``do not have short cycles''.

\begin{definition}
Let $j \geq 1$. A language $L$ is said to be {\emph{without simple cycle of length}} $\leq j$ if the underlying undirected multigraph of the minimal deterministic automaton $\Amin$ for $L$ has no simple cycle of length $k \leq j$.
\end{definition}

\begin{example}  The language $Z_{5}^{1,2}$ has no simple cycle of length $\leq 2$. Indeed, the minimal automaton for $Z_{5}^{1,2}$ is the one depicted on Fig.~\ref{fig:Z5012} with all self-loops removed.
\end{example}

\begin{remark}
The role of the alphabet is crucial. The language $L_1 = Z_{3}^{1} = (\{ 1 \}^{3})^{*} = (111)^{*}$ has no simple cycle of length $\leq 2$ while the language $L_2 = (\{1,2\}^{3})^{*}$ does have simple cycles of length $2$. 

\begin{figure}[!h]
\begin{center}
\includegraphics[scale=0.5]{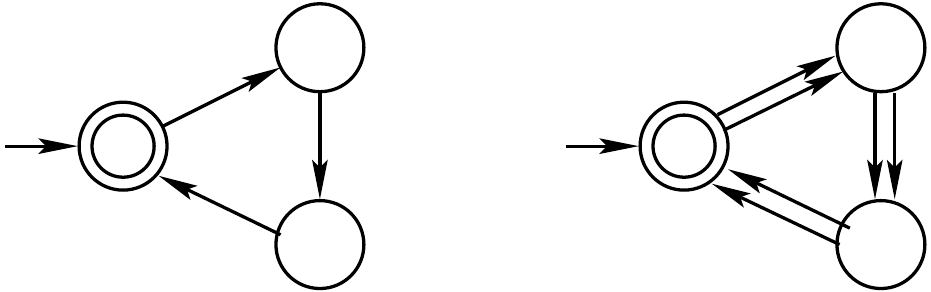}
\end{center}\caption{The minimal automata for $L_1$ and $L_2$. Note that they have the same underlying simple directed graph.}
\end{figure}
\end{remark}

It will be convenient to introduce the following function.

\begin{definition}
Let $m \geq 2$. Set $\rho(m) = \left\{ 
\begin{array}{cl}
3 & {\hbox{if}}\ m \geq 4;\\
4 & {\hbox{if}}\ m = 3;\\
5 & {\hbox{if}}\ m = 2.
\end{array}
\right.$
\end{definition}

\begin{theorem}[Genus estimate] \label{th:genus_estimate}
Let $m \geq 2$. 
If a regular language $L$ on an $m$-letter alphabet has no simple cycle of length $\leq \rho(m)-1$, then
\begin{equation}
1 + \frac{(\rho(m)-2)m-\rho(m)}{2\rho(m)}|L|_{\rm{set}} \leq g(L) \leq 1 + \frac{(m-1)}{2}|L|_{\rm{set}}.
\end{equation}
\end{theorem}

The upper bound is a direct consequence of Euler's formula. The crucial information consists in the lower bound (which is always greater than one). Theorem \ref{th:genus_estimate} generalizes that of \cite[Theorem 8]{BD}. See \S \ref{th:genus_estimate} for the detailed proof. 

\begin{conj}
\label{conj:computable}
The genus $g(L)$ of every regular language $L$ is computable.
\end{conj}

Although the genus of a graph is computable, this conjecture is
far from being obvious. Indeed, given a graph $G$ and a
nonnegative number $g$, there is a procedure, polynomial in time,
that decides whether $G$ embeds into a fixed surface of genus $g$
and if is the case, determines an embedding (not uniquely
determined). The known procedure is linear in the size of the
graph (number of states) and doubly exponential in the genus $g$.
However, this is not enough in order to say anything about the
genus of a {\emph{language}} $L$: it is recognized by an infinity of
deterministic finite automata and since the genus may be realized far away from the minimal deterministic finite automaton $\A_{\rm{min}}$, it is not a priori clear where and when to stop. How far ? According to Theorem \ref{th:exponential_size}, we may need to go after an automaton whose size is at least exponential in the size of $L$. In order to prove the conjecture, one needs a priori bounds that depend on the intrinsic complexity (ideally the size) of the language.

We prove a partial case of the conjecture above.

\begin{definition}
Let $m \geq 2$. Let ${\mathscr{C}}(m)$ be the class of regular languages on $m$ letters without simple cycles of length $\leq \rho(m)-1$.
\end{definition}

%


%

\begin{theorem} \label{th:computable-class}

%
%

Let $m \geq 2$. For each $L \in {\mathscr{C}}(m)$, the topological size $|L|_{\rm{top}}$ and the genus $g(L)$ are computable.
\end{theorem}

\begin{corollary}
The planarity of a regular language $L \in \mathscr{C}(m)$ for $m \geq 2$ is decidable.
\end{corollary}

Since regular languages are (partially) ordered by their genus,
the following finiteness result is useful.

\begin{theorem} \label{th:finiteness}
Let $m \geq 2$. If $\A$ is a deterministic finite automata $\A$ without simple cycles of length $\leq \rho(m)$, then $g(\A) \geq 2$.
Furthermore, for each $g \geq 2$, there is a finite number of deterministic finite automata $\A$ without simple cycles of length $\leq j(m)$ such that $g(\A) = g(L(\A))$.
\end{theorem}

As a corollary, we obtain that there is a finite number of regular languages of fixed genus without simple short cycles.

\begin{corollary}
Let $m \geq 2$. For any $L \in \mathscr{C}(m)$, $g(L) \geq 2$.
Furthermore, for each $g \geq 2$, there is a finite number of regular languages $L \in \mathscr{C}(m)$ such that $g(L) = g$.
\end{corollary}

A few comments may be useful. The hypotheses about the absence of small short cycles and the fixed size of the alphabet is essential. For instance, let $n, p \geq 3$ and consider the language on two letters $$L_{n,p} = \{ w \in \{0,1\}^{*} \ | \ |w|_{0} = 0 \ {\rm{mod}}\ n, \ |w|_{1} = 0 \ {\rm{mod}}\ p \}$$ (where $|w|_{a}$ denotes the number of occurrences of letter $a$ in the word $w$) which can be regarded as the shuffle of $Z_{n}^{1}$ and $Z_{p}^{1}$ \cite[p.65]{Sakarovitch}. 
The minimal automaton for $L_{n,p}$ is obtained as the shuffle product of the minimal automata of $Z_{n}^{1}$ and $Z_{p}^{1}$ respectively. It is not hard to see that this automaton realizes the minimal genus for $L_{n,p}$ -- which is $1$. 

\begin{prop}\label{prop:toric-two-letters}
For $n, p \geq 4$, $L_{n,p}$ is toric.
\end{prop}

We refer to \S \ref{sec:proof-prop-toric-two-letters} for the proof.
It follows from Prop. \ref{prop:toric-two-letters} that there is an infinite family of toric languages on a two-letter alphabet, the simplest example of which has $16$ states. Compare
with \cite{BC} where a two-letter nonplanar language with $35$ states is constructed.


\begin{figure}[!h]
\begin{center}
\includegraphics[scale=0.3]{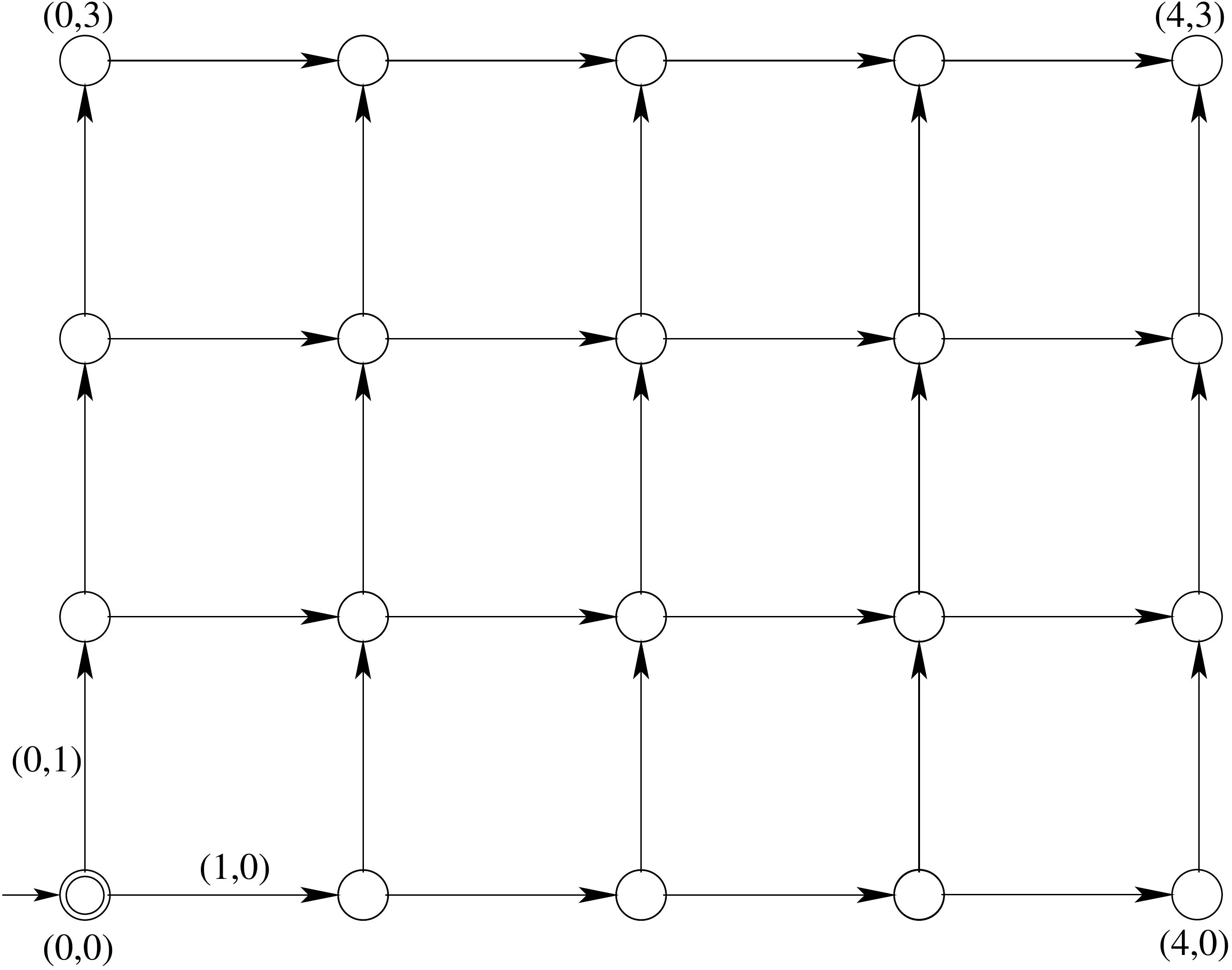}
\end{center}\caption{The minimal automaton for $L_{4,3}$ and its embedding in the torus. The  states $(k,0)$ and $(k,3)$ ($0 \leq k \leq 4$) and the states $(0,l)$ and $(4,l)$ ($0 \leq l \leq 3$) are to be identified as well as the corresponding transitions (the resulting automaton having exactly $12 = 4 \times 3$ states and $24 = 2 \times 12$ transitions) so that the picture represents an embedding of the minimal automaton in the torus. }
\end{figure}

Since $g(L_{n,p}) = 1$ and the lower bound of Theorem \ref{th:genus_estimate} is always greater than $1$, $L_{n,p}$ must have short simple cycles. Indeed, for any $n,p$, the minimal automaton has simple cycles of length $4$, so $L_{n,p} \not\in \mathscr{C}(2)$.

\begin{remark} \label{rem:observations}
Let $2 \leq n_1 \leq n_2 \leq \cdots \leq n_{r}$ be integers. For any finite sequence $(w_{1}, \ldots, w_{s}) \in (\mathbb{Z}/n_1 \mathbb{Z} \times \cdots \mathbb{Z}/n_{r}\mathbb{Z})^{s}$, let
$$Z_{n_1, \ldots, n_r}^{w_1, \ldots, w_s} = \left\{ a_1 \ldots a_k \in \{ w_1, \ldots, w_s \}^{*} \ | \ \sum_{i=1}^{k} a_i = 0 \in \mathbb{Z}/n_1 \mathbb{Z} \times \cdots \mathbb{Z}/n_{r}\mathbb{Z} \right\}.$$
This is a slight generalization of Definition \ref{def:reglang_cyclic_example} where $r = 1$. The language $L_{n,p}$ considered above is also a particular case with $r = 2$:
$L_{n,p} = Z_{n,p}^{(0,1),(1,0)}$. Observe that $Z_{n,p}^{(0,1),(1,0),(1,1)}$ is again a toric language, this time with three letters, and the minimal automaton has simple cycles of length $3$, so it does not belong to $\mathscr{C}(3)$. However, the language $Z_{n,p}^{(0,1),(1,0),(1,1),(-1,1)}$ has four letters and no simple cycles of length $\leq 2$, so by Theorem \ref{th:genus_estimate}, its genus is bounded below by $1 + \frac{1}{6}np$. 

\begin{figure}[!h]
\begin{center}
\includegraphics[scale=0.2]{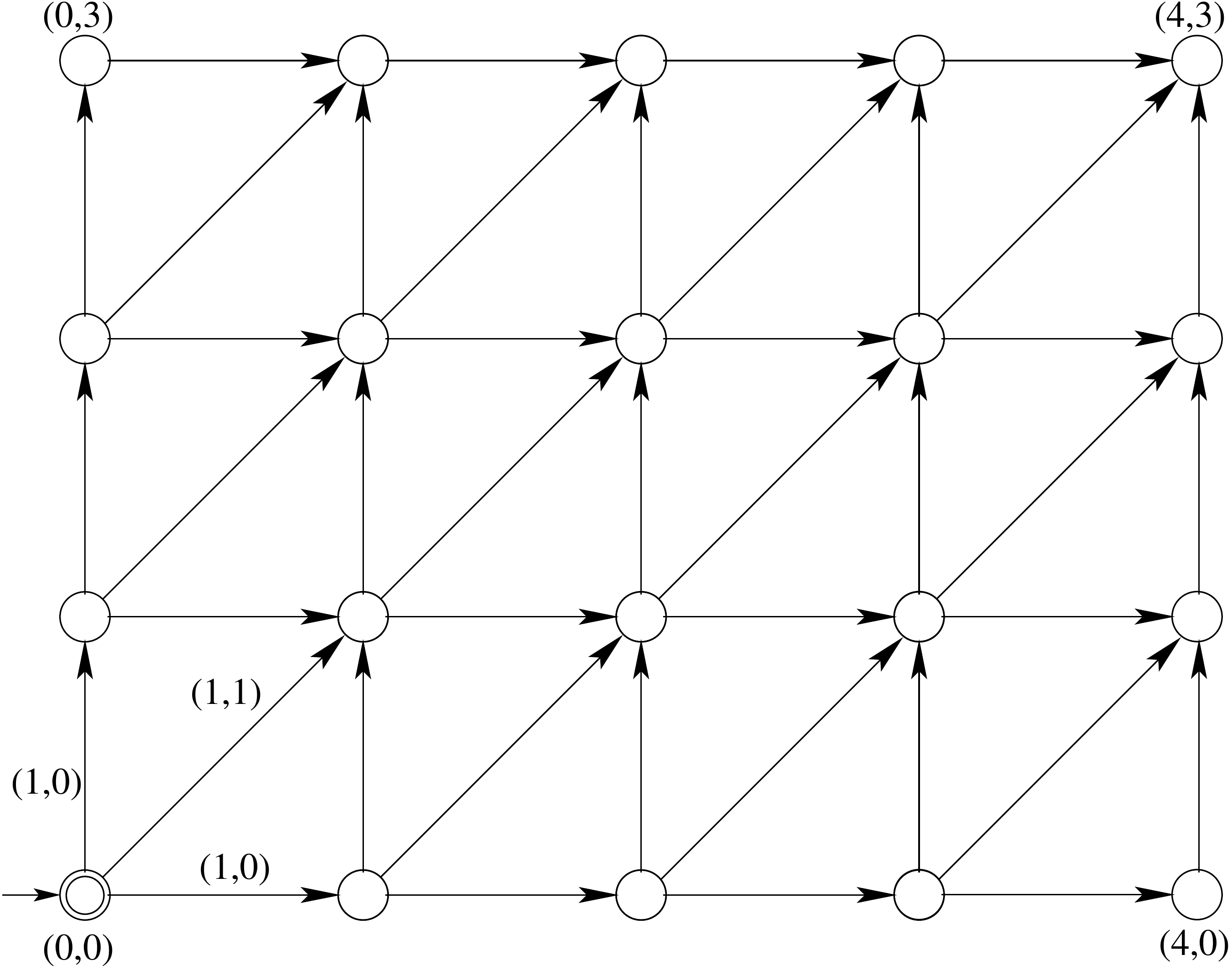} \ \includegraphics[scale=0.2]{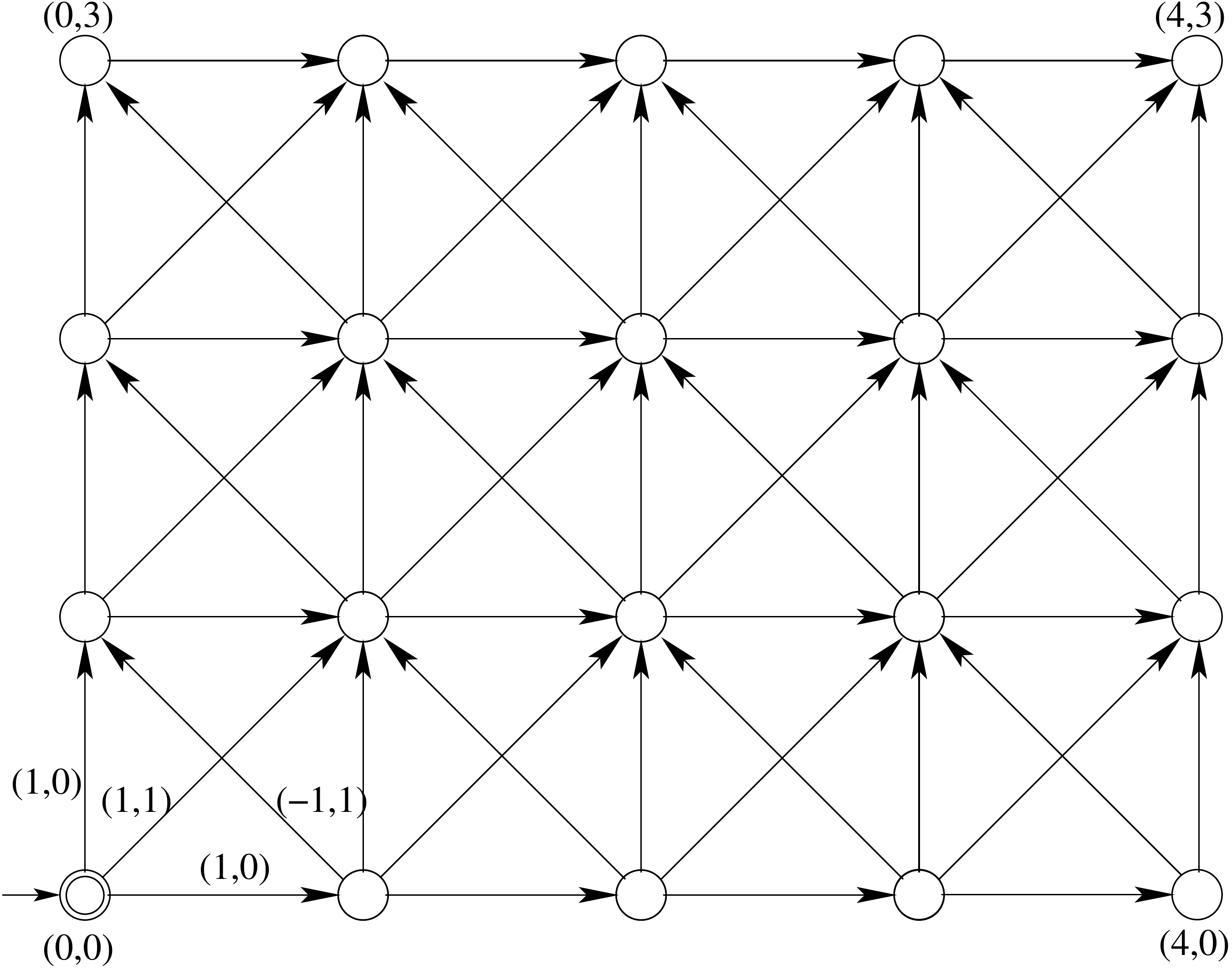}
\end{center}\caption{The minimal automata for $Z_{4,3}^{(1,0),(0,1),(1,1)}$ and $Z_{4,3}^{(1,0),(0,1),(1,1),(-1,1)}$ respectively (with the same identification convention as in the previous figure). The first one embeds into the torus, the second one (nor any deterministic automaton equivalent to it) does not.}
\end{figure}

This provides another example of hierarchy based on the genus. Note that contrary to the genus, the syntactic monoid does not distinguish between the languages $Z_{n,p}^{(0,1),(1,0),(1,1),(-1,1)}$ and $Z_{n,p}^{(0,1),(1,0),(1,1),(-1,1)}$, since
 $${\mathfrak{M}}(Z_{n,p}^{(0,1),(1,0),(1,1)}) = {\mathfrak{M}}(Z_{n,p}^{(0,1),(1,0),(-1,1)}) = \mathbb{Z}/n\mathbb{Z} \times \mathbb{Z}/p\mathbb{Z}.$$


\end{remark}

Another observation is that given a language $L \in \mathscr{C}(m)$, it is easy to build an infinite number of languages of the same genus $g(L)$, but of course the produced languages will have short simple cycles. For instance, if $A$ denotes the alphabet of $L$ and has at least two letters, then for any $k \geq 0$, $g(A^{k}\cdot L) = g(L)$. Indeed, given a genus-minimal automaton for $L$, an automaton of the same genus can be built for the composition $A \cdot L$: it is easily seen to have one simple cycle of length $2$, see Fig. \ref{fig:modified_language_with_same_genus}.

\begin{figure}[!h]
\begin{center}
\begin{picture}(0,0)%
\includegraphics{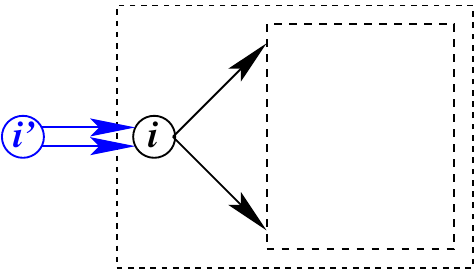}%
\end{picture}%
\setlength{\unitlength}{2368sp}%
\begingroup\makeatletter\ifx\SetFigFont\undefined%
\gdef\SetFigFont#1#2#3#4#5{%
  \reset@font\fontsize{#1}{#2pt}%
  \fontfamily{#3}\fontseries{#4}\fontshape{#5}%
  \selectfont}%
\fi\endgroup%
\begin{picture}(3805,2144)(-32,-833)
\put(976,-736){\makebox(0,0)[lb]{\smash{{\SetFigFont{14}{16.8}{\familydefault}{\mddefault}{\updefault}{\color[rgb]{0,0,0}$ $}%
}}}}
\end{picture}%
\end{center} \caption{A genus-minimal automaton for $L$ (on two letters) with initial state $i$; the corresponding genus-minimal automaton for the language $A \cdot L$ has one extra state $i'$ and two extra transitions from $i'$ to $i$, hence one simple cycle of length $2$.}
\label{fig:modified_language_with_same_genus}
\end{figure}
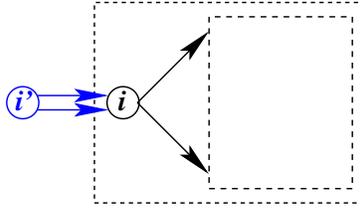




\section{Directed emulators} \label{sec:diem}


In this section, we give a graph-theoretical approach to
the study of the genus of regular languages. A \emph{directed graph} (or a \emph{digraph}) $G$ consists of a set $V$ of {\emph{vertices}} and a set $E$ of {\emph{edges}} and two maps $s,t: E {\rightrightarrows} V$ (resp. ``source'' and ``target''). A {\emph{morphism}} $G \to H$ {\emph{between directed graphs}} $G = (V_G, E_G,s,t)$ and $H = (V_H, E_H,s',t')$ is a pair of maps $(f_{V}, f_{E})$ where $f_{V}:V_{G} \to V_{H}$ is a map between the set of vertices of $G$ to the set of vertices of $H$ and $f_{E}:E_{G} \to E_{H}$ is a map between the set of edges of $G$ to the set of edges of $H$, that preserves the adjacency relation: $f_{V} \circ s = s' \circ f_{E}$ and $f_{V} \circ t = t' \circ f_{E}$. A digraph $G = (E,V,s,t)$ is \emph{simple} if for any $v, w \in V$, there is at most one edge $e \in E$ such that $s(e)=v$ and $t(e) = w$ and there is no edge $e \in E$ such that $s(e) = t(e)$ (no self-loop). 


In order to investigate the relation between automata and graphs,
we refine a notion that was defined by M. Fellows in the context of undirected graphs.

\begin{definition}
Let $G = (V,E,s,t)$ be a digraph. We say that a digraph $G'=(V',E',s',t')$ is a {\emph{directed emulator}} of $G$ if there is a surjective map $p:V' \to V$ such that for any edge $e \in E$ and any $x \in p^{-1}(s(e))$, there is an edge $e' \in E'$ such that $s(e') = x$ and $t(e') \in p^{-1}(t(e))$. Such a map $p$ is called a {\emph{directed emulator map}} and we say that the digraph  
$G = (V,E)$ is a {\emph{directed amalgamation}} of $G' = (V',E')$ 
if $G'$ is a directed emulator of $G$.
\end{definition}

If we regard the digraph as a simplicial $1$-complex, 
a directed emulator map is a simplicial map mapping the outgoing edges from each vertex $x' \in V'$ surjectively onto the outgoing edges from the image vertex $x \in V$.

If $E'$ is nonempty, then a directed emulator map induces a surjective map on the vertices $q:E' \to E$ that preserves the adjacency relation. Therefore, a directed emulator map is a surjective morphism of digraphs. (if E' nonempty, otherwise add self loops).

\begin{remark} 
\label{rem:edge-contraction-is-case-of-amalgamation}
The endpoints $s'(e')$ and $t'(e')$ of a given edge $e' \in E'$ may happen to be sent by a directed emulator map to a single vertex in $V$ (provided the local condition at the vertex is respected). 
In particular, an edge contraction induces a directed emulator map.
\end{remark}


\begin{remark}
The definition of a direct emulator is a weakening of the definition of a directed graph covering. A covering map
maps the outgoing edges from each vertex $x' \in V'$ bijectively onto the outgoing edges from the image vertex $x \in V$. A graph covering map is a special kind of emulator map. An emulator map is a special kind of
surjective morphism of digraphs.
\end{remark}

\begin{remark} \label{rem:trivial-digraph}
The digraph that consists of one vertex and no edge is the directed amalgamation of any nonempty digraph.
\end{remark}


A {\emph{morphism}} $\A \to \B$ {\emph{between automata}} is a map $f:Q_{\A} \to Q_{\B}$ from the set of states of $\A$ to the set of states of $\B$ with the following properties: 
\begin{enumerate}
\item[(1)] $f$ sends the initial state of $\A$ to the initial state of $\B$;
\item[(2)] $f$ sends the set of final states 
of $\A$ into the set of final states of $\B$;
\item[(3)] The following diagram commutes:
$$ \xymatrix{
Q_{\A} \times {\mathcal{A}} \ar[r] \ar[d]_-{f \times {\rm{id}}} & Q_{\A} \ar[d]_{f} \\
Q_{\B} \times {\mathcal{A}} \ar[r] & Q_{\B}
}$$
where ${\mathcal{A}}$ denotes the alphabet and the horizontal maps are the transition maps of $\A$ and $\B$ respectively.
\end{enumerate}


Deterministic finite accessible and co-accessible automata on a fixed finite alphabet with their morphisms form a category DFA$_0$. (See e.g., \cite[III.4]{Eilenberg}.) Minimal deterministic finite automata are final objects of the category DFA$_{0}$.
We investigate more closely this category and the related category of directed graphs. 


There is a forgetful functor ${\widetilde{\mathscr{G}}}$ from the category DFA$_{0}$ to the category DiGr of finite digraphs: ${\widetilde{\mathscr{G}}}$ forgets the labels on the transitions and the distinguished states.


\begin{lemma} \label{lem:cat}
The functor $\widetilde{\mathscr{G}}$ is full and preserves the genus of objects.
\end{lemma}


In particular, a regular language $L$ gives rise, via its minimal automaton $A_{\rm{min}}(L)$, to a digraph denoted $\widetilde{G}(L)$.


%
%

\begin{definition}
The category ${\rm{DFA}}$ is defined as follows.

- An object in DFA is a morphism in DFA$_0$, i.e. a morphism
$\A' \to \A$ between deterministic finite (accessible, co-accessible) automata.

- A morphism in DFA is a commutative diagram of automata morphisms:
\begin{equation} \xymatrix{
\A' \ar[r] \ar[d] & \B' \ar[d] \\
\A \ar[r] & \B
}\label{eq:morphism-in-dfa}
\end{equation}
\end{definition}

\begin{definition}
The category DiEm of {\emph{directed emulators}} is defined as follows. An object in DiEm is
a  directed emulator map. A morphism in DiEm is a commutative
diagram
\begin{equation} 
\xymatrix{
G' \ar[r] \ar[d]_{\pi_{G}} & H' \ar[d]^{\pi_{H}} \\
G \ar[r] & H
} \label{eq:morphism-in-diem}
\end{equation}
where the vertical maps are directed emulators and the horizontal maps are digraph morphisms.
\end{definition}

\begin{prop} \label{prop:functor-automata-emulator}
The functor $\widetilde{\mathscr{G}}:{\rm{DFA}}_{0} \to {\rm{DiGr}}$ induces a full functor ${\rm{DFA}} \to {\rm{DiEm}}$.
\end{prop}
We still denote $\widetilde{\mathscr{G}}$ the functor ${\rm{DFA}} \to {\rm{DiEm}}$.


\begin{proof}
We really only need to verify that the functor sends a morphism $p: \A' \to \A$ between automata to a directed emulator graph. 
Consider two distinct states $x, y \in \A$ such that $\delta_{\A}(x,l) = y$ for some letter $a$. Let a state $x' \in p^{-1}(x)$. By property (3) of the definition, $\delta_{\A'}(x',a)$ lies in $p_{\A}^{-1}(y)$. Therefore there is a transition from $x'$ to some $y' \in p_{A}^{-1}(y)$ in $\A$. This implies that the induced map
${\mathscr{G}}(\A') \to {\mathscr{G}}(\A)$ is
a directed emulator map.
\end{proof}

\begin{definition}
Let $\A \in {\rm{DFA}}_0$. Let ${\rm{DFA}}(\A)$ be the set of all automata $\B \in {\rm{DFA}}$ such that there is at least one morphism $\B \to \A$.
\end{definition}

\begin{remark}
If $\A$ is a minimal automaton, then there is at most one morphism $\B \to \A$. Furthermore, if such a morphism exists, then it is the canonical projection induced by the map sending a state to its equivalence class.
\end{remark}

\begin{lemma}
Given $\A_1, \A_2 \in {\rm{DFA}}(\A_{\rm{min}})$, there exists $\A_{12} \in {\rm{DFA}}$ such that the following diagram is commutative
$$
\xymatrix{
& \A_{12} \ar@{.>}[ld]_{p_{121}} \ar@{.>}[dd]^{p_{12}} \ar@{.>}[rd]^{p_{122}} & \\
\A_{1} \ar[rd]_{p_1} & & \A_2 \ar[ld]^{p_{2}} \\
& \Amin &
}
$$
\end{lemma}
\begin{proof}
Take $\A_{12}$ to be the fibered product of $\A_1$ and $\A_2$ over $p_{1} \times p_{2}$.
\end{proof}

\begin{definition}
Let $G$ be a digraph.
We denote DiEm$(G)$ the set of directed emulators of $G$.
\end{definition}

\begin{lemma} \label{lem:common-directed-emulator}
Given $G_1, G_2 \in {\rm{DiEm}}(G)$, there exists $G_{12} \in  {\rm{DiEm}}(G)$ making the following diagram
commute.
$$
\xymatrix{
& G_{12} \ar@{.>}[ld]_{} \ar@{.>}[dd]^{} \ar@{.>}[rd]^{} & \\
G_{1} \ar[rd]_{} & & G_2 \ar[ld]^{} \\
& G &
}
$$
where each map is a directed emulator map.
\end{lemma}


\begin{proof}
Apply the functor $\widetilde{\mathscr{G}}$ to the diagram of the previous lemma.
\end{proof}

\begin{lemma} \label{lem:absolute-common-directed-emulator}
Any two digraphs have a common directed emulator.
\end{lemma}

\begin{proof}
They both are directed emulators of the digraph that consists of one vertex with no edge (Remark~\ref{rem:trivial-digraph}).
\end{proof}

\begin{definition}
The category of {\emph{simple directed graphs}} is denoted ${\rm{SDiGr}}$. 
\end{definition}

\begin{lemma} \label{lem:full_from_simple_to_directed}
The category ${\rm{SDiGr}}$ of simple directed graphs is a full subcategory ${\rm{DiGr}}$ of directed graphs. The forgetful functor $R:{\rm{DiGr}} \to {\rm{SDiGr}}$ that consists in forgetting self-loops and merging  multiple oriented edges into one is a full retraction functor of the inclusion functor ${\rm{SDiGr}} \to {\rm{DiGr}}$. Both preserve genus.
\end{lemma}

\begin{example} Fig. \ref{fig:from_multiple_to_simple} depicts the action of the forgetful map on a directed graph.
\begin{figure}[!h]
\begin{center}
\includegraphics[scale=0.5]{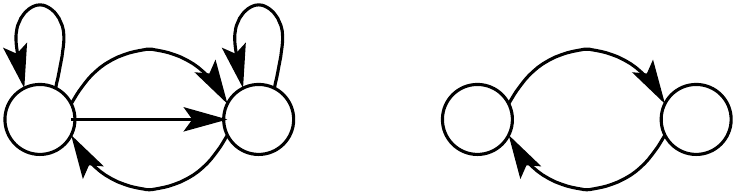}
\end{center}
\caption{Action of the retraction map on a directed graph.} \label{fig:from_multiple_to_simple}
\end{figure}
\end{example}

\begin{definition}
The category ${\rm{SDiEM}}$ of {\emph{simple directed emulators}} is the category whose objects are directed emulators between simple directed graphs and whose morphisms are commutative squares where vertical maps are directed emulators and horizontal maps are morphisms of (simple) directed graphs.
\end{definition}

\begin{lemma} \label{lem:full_from_automata_to_directed_emulator}
The category ${\rm{SDiEM}}$ of {\emph{simple directed emulators}} is a full subcategory of the category ${\rm{DiEm}}$ of directed emulators.
\end{lemma}

\begin{corollary}
The assignment ${\mathscr{G}}: R \circ {\widetilde{\mathscr{G}}}: {\rm{DFA}} \to {\rm{SDiGr}}$ induced by the map sending a finite deterministic automaton its underlying simple directed graph is a full and genus preserving, functor. 
\end{corollary}

\begin{corollary}
For a regular language $L$,
$$ g(L) = \min \{ g(R(\A)) \ | \ \A \in {\rm{DFA}}, \ L(\A) = L \}.$$
\end{corollary}

\begin{remark}
In the final computation of the genus $g(R(\A))$, one can disregard the orientation of the edges. The subtlety is that the minimum runs over \emph{all} deterministic automata such that $L(\A) = L$. Hence the orientation can be dropped only once a specific representative automaton is found.  
\end{remark}

\begin{definition}
Let $G$ be a simple digraph.
We denote SDiEm$(G)$ the set of simple directed emulators of $G$.
\end{definition}

\begin{lemma}
Lemmas $\ref{lem:common-directed-emulator}$ and $\ref{lem:absolute-common-directed-emulator}$ remain valid when replacing ``digraph'' by
``simple digraph'' and ``directed emulator'' by ``simple directed emulator''.
\end{lemma}

\begin{definition}
Let $L$ be a regular language. The {\emph{underlying directed graph}} $G(L)$ of $L$ is the simple directed graph associated to the minimal automaton $\A_{\rm{min}}(L)$ of $A$. 
\end{definition}

The following result is the main observation of the section. 

\begin{theorem} \label{th:genus-and-emulator}
Let $L$ be a regular language. The following assertions are equivalent:
\begin{enumerate}
\item[$(1)$] The language $L$ has genus $ \leq g$.
\item[$(2)$] The associated digraph $G(L)$ has a directed emulator of genus $\leq g$.
\item[$(3)$] The associated digraph $G(L)$ has a simple directed emulator of genus $\leq g$.
\end{enumerate}
\end{theorem}

The equivalence between $(2)$ and $(3)$ follows from Lemma \ref{lem:full_from_simple_to_directed}. 

\begin{corollary}
The genus of a regular language $L$ is equal to the minimum of all genera of simple directed emulators over $G(L)$:
$$ g(L) = \min \{ g(G) \ | \ G \in {\rm{DiEM}}(G(L)) \} = \min \{ g(G) \ | \ G \in {\rm{SDiEM}}(G(L)) \}$$
\end{corollary}

%

Theorem \ref{th:genus-and-emulator} allows to translate
questions about the genus of languages into questions 
about directed emulators of simple digraphs (and vice-versa). 

\begin{corollary}
If two languages $L$ and $L'$ have the same underlying directed graphs $G(L)$ and $G(L')$, then $g(L) = g(L')$.
\end{corollary}

\begin{corollary}
Let $A$ be a deterministic automaton and $L$ be the language computed by $A$. Let $A'$ be a deterministic automaton obtained from $A$ by the following operations:

\begin{figure}[!htb]
    \centering
    \begin{minipage}{.5\textwidth}
        \centering
        \includegraphics[scale=0.5]{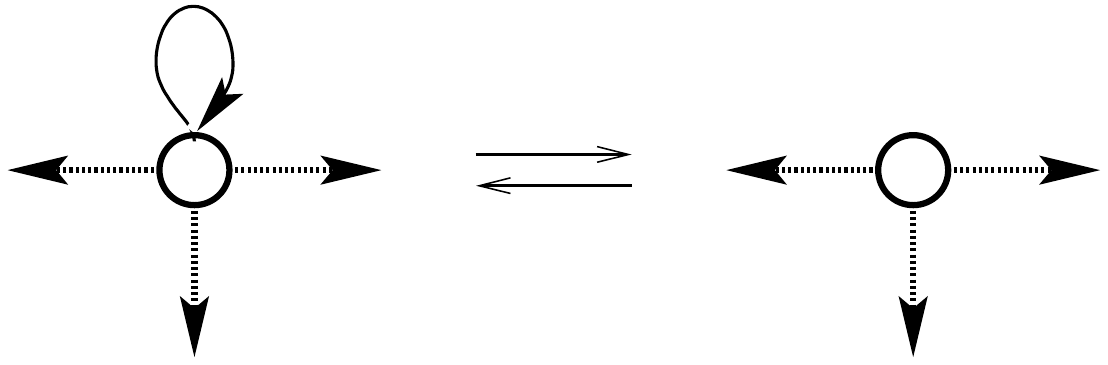}
        \caption{Adding or removing a self-loop}
        \label{fig:directed-cycle}
    \end{minipage}%
    \begin{minipage}{0.5\textwidth}
        \centering
        \includegraphics[scale=0.5]{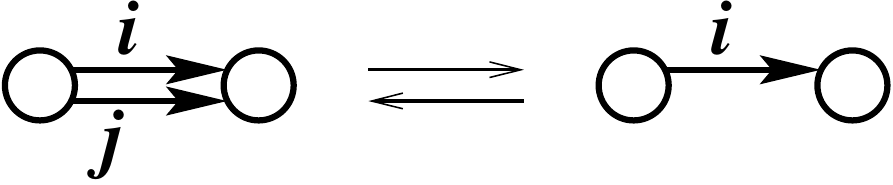}
        \caption{Adding or removing one extra transition between states}
        \label{fig:edge-contraction}
    \end{minipage}
\end{figure}
Let $L'$ be the language computed by $A'$. Then $g(L) = g(L')$.
\end{corollary}

Caution needs to be exercised to apply this corollary since some of the operations (those adding transitions) do not preserve a priori determinism.

\begin{proof}
The operations do not affect the underlying directed graph of the respective minimal automata of $A$ and $A'$, so $G(L) = G(L')$.
\end{proof}

\section{The proofs} \label{sec:proofs}

%

\subsection{Proof of Theorem \ref{th:new-hierarchy-example} (A new explicit example of genus hierarchy with exact genus formula)}  \label{subsec:new-hierarchy-proof}

The language $Z_{2k+1}^{1,2,\ldots, k}$ is computed by
the following automaton, denoted $A = A_{2k+1}^{1,2,\ldots, k}$. The set of states is $Q = \mathbb{Z}/(2k+1)\mathbb{Z}$, with initial and final state $0$. The transitions are given by the rule $i\  \overset{j}{\to} i+j$ for $i \in Q$ and $j \in \{1, 2, \ldots, k\} \subset \mathbb{Z}/(2k+1)\mathbb{Z}$. It is readily observed that $A$ is the minimal automaton. The underlying unoriented multigraph is the complete graph $K_{2k+1}$. We verify two properties:

- $K_{2k+1}$ has no self-loop and has no simple cycle of length $2$ (the minimal length of a simple cycle
is $3$)

- The cardinality of the alphabet is $k \geq 4$.

According to Theorem~\ref{th:genus_estimate} (see also \cite[Th.~8]{BD}), $g(Z_{2k+1}^{1, 2, \ldots, k}) \geq 1 + \frac{(k-3)(2k+1)}{6}$. To prove that this lower bound for the genus is actually an equality, we notice that the genus of the minimal automaton provides an upper bound. So
$$ 1 + \frac{(k-3)(2k+1)}{6} \leq g(Z_{2k+1}^{1,2,\ldots, k}) \leq g(A) = g(K_{2k+1}) = 
\left\lceil \frac{(2k-2)(2k-3)}{12} \right\rceil. $$
The last equality is the exact formula for the genus
of the complete graph on $2k+1$ vertices.
It remains to observe that the ceiling function of the lower bound is exactly the upper bound. This is the desired result.

\subsection{Proof of Theorem \ref{th:hierarchy-2-letters} (Existence of a genus hierarchy for languages on a $2$-letter alphabet)}

Let $A = \mathbb{Z}/2{\mathbb{Z}}$ be the alphabet. For $k \geq 5$, consider the finite deterministic automaton $\A_{k}$ defined as follows. The set of states is $Q_{k} = \mathbb{Z}/6\mathbb{Z} \times \mathbb{Z}/k\mathbb{Z}$. The transitions are $$(i,j) \overset{0}{\to} (i+1, j), \ \ (i,j) \overset{1}{\to} (2i, j+1).$$ Pick the state $(0,0)$ as the initial and unique final state. See Fig. \ref{fig:two-letters-hierarchy} for a picture of the automaton $\A_{k}$.

\begin{figure}[!h]
\includegraphics[scale=0.3]{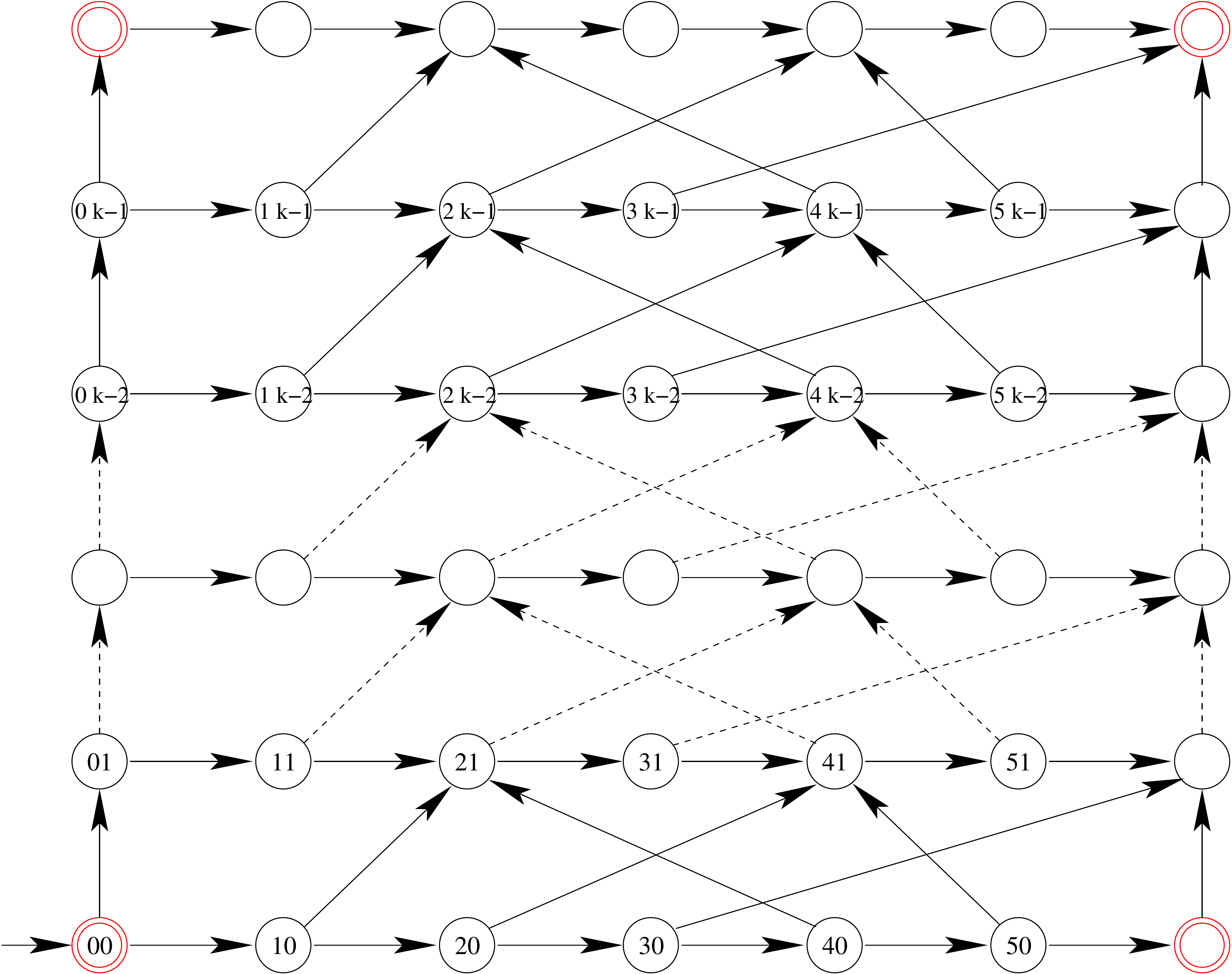}
\caption{The automaton $\A_{k}$ drawn (with crossings) on a torus: the states $(j,0)$ and $(j,k)$ for $0 \leq j \leq 6$ (resp. the states $(0,l)$ and $(6,l)$ for $0 \leq l \leq 6$) are to be identified, as well as the corresponding transitions.}\label{fig:two-letters-hierarchy}
\end{figure}

It is easily seen that $\A_k$ is deterministic, complete and minimal. It is readily 
verified that $\A_{k}$ has no simple cycle of length less than or equal to $4$. Therefore Theorem \ref{th:genus_estimate} applies: the language $L_{k}$ recognized by $\A_{k}$ has genus $g(L) \geq 1 + \frac{3k}{20}$. This implies the desired result.

\subsection{Proof of Theorem \ref{th:exponential_size} (Planar regular languages with exponential topological size) }  \label{sec:exponential_size}

On the alphabet $\mathbb{Z}/5 \mathbb{Z}$, given $n\geq 0$, let us consider the automaton $A_n = (Q_n, i_n, F_n, \delta_n)$ defined as follows. The set of states is $Q_n = \mathbb{Z}/5 \mathbb{Z} \times \{0, \cdots, n\} \cup \{p_0, \top, \bot\}$. The initial state is $p_0$, there is a unique final state $\top$. For all $a, b \in \mathbb{Z}/5 \mathbb{Z}$, let $\delta_n(p_0, a) = (a,0)$, $\delta_n((a, n),a) = \top$, if $a \neq b$, $\delta_n((a,n),b) = \bot$ and for $j<n$, $\delta_n((a,j),b) = (a+b,j+1)$. Its corresponding language is $L_n = \{ a_0\cdots a_{n+1} \mid \sum_{i = 0,n} a_i = a_{n+1} \}$.  

It is straightforward that all states of $A_n$ are accessible and that $A_n$ is minimal, its states being non equivalent. The language $L_n$ is finite, thus planar. Indeed, one may span the complete tree of depth $n+2$ to describe the language which has thus topological size smaller than $5^{n+2}$. Let us suppose that $B_n = (R_n, j_n, G_n, \eta_n)$ is a minimal planar automaton recognizing $L_n$. Without loss of generality, we can suppose that its states have the shape $(s,t)$ with $s \in Q_n$ and $t \in T$, that is $\pi : (s,t) \mapsto s$ defines the projection on the minimal automaton. 

We qualify states of the shape $(a,j,t)$ with $j<n$ to be internal states. 
For any internal state $s = (a,j,t)$, the transition function $\eta_n(s, \cdot) : \mathbb{Z}/5 \mathbb{Z} \to R_n$ is injective, because $\delta_n=\pi\circ\eta_n$ is injective. Explicitly, for any $b\neq c\in\mathbb{Z}/5 \mathbb{Z}$, we have $\eta_n(s,b)\neq\eta_n(s, c)$.


Let $G_n = {\widetilde{{\mathscr{G}}}}(B_n)$ be the underlying (planar) graph of $B_n$. 
Given $j \in \{0, \ldots, n-1\}$, let $S_j$ be the subgraph of $G_n$ where any vertices outside $\mathbb{Z}/5 \mathbb{Z} \times \{ j, j+1 \} \times T$ have been removed with their incoming and outcoming edges. Being a subgraph of $G_n$, the graph $S_j$ is planar.  We denote $K$ (respectively $M$) the set of states of $B_n$ of the shape 
$(a, j, t)$  (resp. $(a, j+1, t)$) and $k = |K|$ (resp. $m = |M|$).  

Any state $s \in K$ is internal. We have seen above that $\eta_n(s,\cdot)$ is injective. Thus, there are exactly $5$ outgoing edges from state $s$, each of which pointing to a different state. Two partial conclusions. First, let $e$ be the number of edges in $S_j$, we have $e = 5k$. Second, there are no bigons in $S_j$: none of the patterns $s \to s' \to s$ or $s\to s'\leftarrow s$ can happen.

Let $f$ be the number of faces in $S_j$. Euler's formula for planar graphs applied in $S_j$ gives us $k + m + f = 5 k + 2$, that we can rewrite:
:
\begin{equation} m + f = 4k + 2.\label{eq-jjj}
\end{equation}

Let $f_i$ be the number of $i$-gon in $S_j$.  Thus, $f = f_1 + f_2 + \cdots$. Observe that due to the definition of $B_n$, there are neither simple odd polygons (that is a $2i+1$-gon for $i \in \N$), neither bigons as justified above. Thus, $f = f_4 + f_6 + \cdots$. A simple counting argument shows that $2 \times e = 4 f_4 + 6 f_6 + \cdots=10k$. In other words, $\cfrac{5k}{2}= f_4 + \cfrac{6}{4} f_6 + \cdots \geq f_4 + f_6 + \cdots = f$.  By relation~(\ref{eq-jjj}), we get 
\begin{equation}
m = 4 k + 2 - f \geq \cfrac{3k}{2}+2 \geq \cfrac{3k}{2} \label{eq_induction}
\end{equation}
Take $K=3/2$. Denote by $N_j$ the states in layer $j$, that is of the shape $(a,j,t)$, and by $n_j$ the cardinal of $N_j$. By induction on $j\geq 0$, we prove $n_j \geq 5 \times (3/2)^{j}$ for $j \leq n$. For the base case, observe that there are at least $5$ states in each layer (there are $5$ in the minimal automaton). The induction step is a direct consequence of the inequality~(\ref{eq_induction}). The result follows.


\subsection{Proof of Theorem \ref{th:genus_estimate} (Genus estimate)}

We need to prove the stated lower bound. Given an integer $k \geq 1$ and a cellular embedding of a graph in a surface, we let $f_{k}$ denote the number of faces of length $k$.

Set $$A(j) = \displaystyle \sum_{k \geq j} \frac{k(m-1)-2m}{4m} f_{k}, \ \ \ B(j) = \sum_{k \geq j} k \, f_{k}.$$

Then
$$ A(j) \geq \left( \frac{m-1}{4m} - \frac{1}{2j} \right) \, B(j).$$
Let $\A$ be a complete minimal genus finite deterministic automaton recognizing $L$. 
By \cite[Th. 5]{BD}, $g(\A) = 1 + A(1)$. By hypothesis,
$\Amin$ has no simple cycle of length less or equal to $j-1$. It follows from Lemma \ref{lem:nosimplecycle} below (\S \ref{sec:cycles}) that $\A$ has no simple cycle of length less or equal to $j-1$.
Consider a minimal embedding (hence cellular) of $\A$ into a genus $g(L)$ oriented closed surface $\Sigma$. Consider now a face $f$ in $\Sigma$. 
 If the length of the face is less or equal to $4$, then any cycle $c$ of $\A$ bounding $f$ must be simple.  We deduce that there is no face of length less than $j-1$: $f_{1} = \cdots = f_{j-1} = 0$. Hence

\begin{align*}
g(\A)  = 1 + A(j)  & \geq  1+ \left( \frac{m-1}{4m} - \frac{1}{2j} \right) \, B(j) \\
& \geq 1 + \left( \frac{m-1}{4m} - \frac{1}{2j} \right)\, B(1) \\
& = 1+ \left( \frac{m-1}{4m} - \frac{1}{2j} \right)\, 2m|\A|
\end{align*}
Therefore
\begin{equation}
g(\A) \geq 1 + \frac{(j-2)m-j}{2j}|\A|. \label{eq:key_genus_ineq}
\end{equation}
Since $|\A| \geq |\Amin| = |L|_{\rm{set}}$, we deduce the desired result.


\begin{remark} \label{rem:inequality-for-automata}
The inequality $(\ref{eq:key_genus_ineq})$ holds for any complete minimal genus deterministic automaton recognizing $L$ under the hypotheses of Theorem \ref{th:genus_estimate}. It is in general sharper than the lower bound of the theorem.
\end{remark}

\begin{remark}
There are simple graphs supporting nonsimple bounding closed walks (see Fig. \ref{fig:special_graphs}): the complement in the $2$-sphere of each of these embedded graphs is an open cell whose boundary is a nonsimple closed walk in the graph.
It is left to the reader to verify that if one of the graphs below is the underlying simple graph of a $2$-letter automaton $\A$, then the underlying multigraph of $\A$ has a {\emph{simple}} cycle of length less or equal to $4$.
\begin{figure}[!h]
\includegraphics[width=0.25\textwidth]{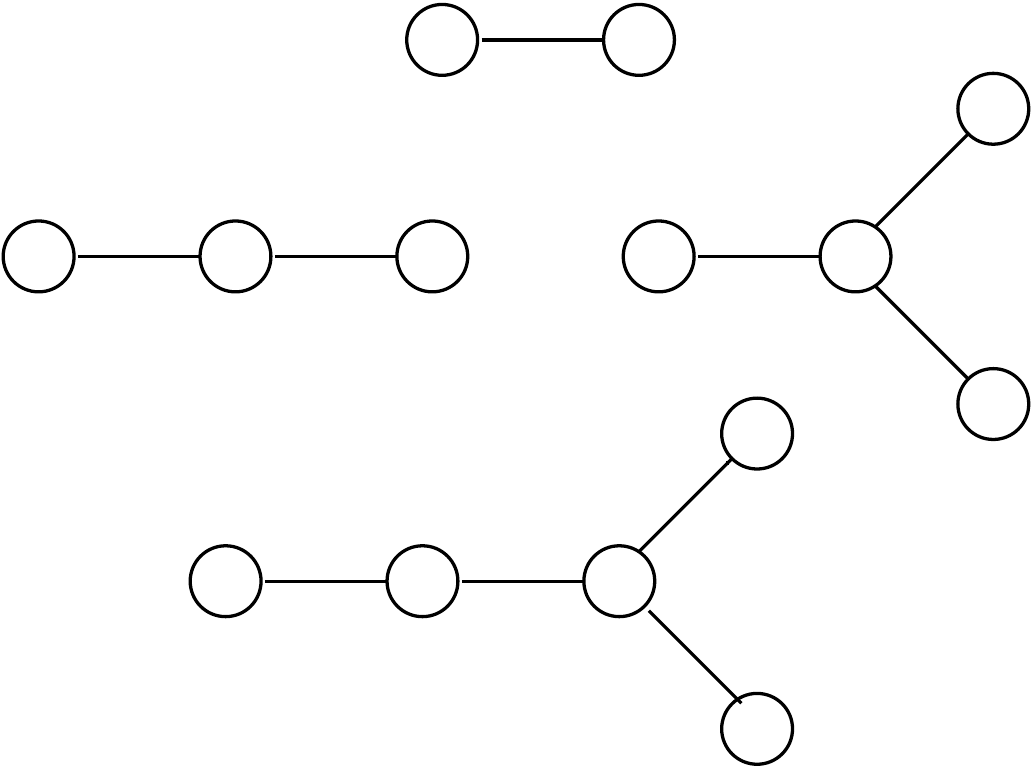}
\caption{\label{fig:special_graphs} Planar graphs supporting nonsimple bounding cycles.}
\end{figure}
\end{remark}

\subsection{Proof of Theorem \ref{th:computable-class} (Computability of genus)}


%

Let $\A$ be a deterministic finite automaton such that $L(\A) = L$ and $g(\A) = g(L)$. Let $\Amin$ be
the minimal deterministic automaton for $L$.
By Lemma \ref{lem:nosimplecycle}, since $\Amin$ has no simple cycles of length $\leq j-1$, neither has $\A$. According to the inequality $(\ref{eq:key_genus_ineq})$ and \cite[Prop. 2]{BD},
$$ 1 + \left( \frac{(j-2)m-j}{2j} \right)\, |\A| \leq g(A) = g(L) \leq g(\Amin) \leq 1 + \left( \frac{m-1}{2} \right) |\Amin|.$$
Since $|L|_{\rm{set}} = |\Amin| \leq |\A|$, we have
$$ 1 + \left( \frac{(j-2)m-j}{2j} \right)\, |L|_{\rm{set}} \leq g(\A) \leq 1 + \left( \frac{m-1}{2} \right) |L|_{\rm{set}}.$$
The set $$E = \left\{ n \in \mathbb{N} \ | \ |L|_{\rm{set}} \leq n, \ 1 + \left( \frac{(j-2)m-j}{2j}\ n \right) \leq g(\Amin) \right\}$$ is finite and contains $|L|_{\rm{top}}$. Let $n \in E$. There is only a finite number of DFAs of
fixed size $n$, hence a finite number of DFAs of size $n$ and computing $L$. Therefore the set
$$ F = \{ \A \in {\rm{DFA}}(\Amin) \ | \ L(\A) = L, \ |\A| \in E \} $$ is finite and contains every deterministic finite automaton computing $L$ of minimal genus. Now for each individual automaton $\A \in F$, its genus is computable (computation of the genus of a graph). The minimum of the finite list of genera thus computed is the genus of $L$.

More generally, given a finite graph $\A$, there is a known algorithm to construct all embeddings of $\A$ into a surface of minimal genus. It follows from the argument above that there is 
an algorithm to construct every deterministic finite automaton $\A$ such that $L(\A) =  L$ and
$g(\A) = g(L)$. In particular, there is only a finite number of them $\A_1, \ldots, \A_r$. In particular, one can compute $|L|_{\rm{top}} = \min \{ |\A_1|, \ldots, |\A_{r}| \}$. This completes the proof.

%
%

\subsection{Proof of Theorem \ref{th:finiteness} (Finiteness of minimal fixed genus automata)}

Let $m \geq 2$. According to $(\ref{eq:key_genus_ineq})$ (see Remark \ref{rem:inequality-for-automata}), $g(\A) > 1$ for any deterministic finite automaton $\A$ without simple cycles of length $\leq j-1$. It follows that there is no deterministic finite automaton $\A$ without simple cycles of length $\leq j-1$ such that $g(\A) \in \{0, 1\}$. In particular, there is no language $L \in {\mathscr{C}}_{j}(m)$
of genus $0$ or $1$.

Let $g \geq 2$. Let $\A$ be any deterministic finite automaton $\A$ without simple cycles of length $\leq j-1$, such that $g(\A) = g(L(\A)) = g$. According to  $(\ref{eq:key_genus_ineq})$,
$$ 1 + \frac{(j-2)m-j}{2j}\, |\A| \leq g. $$
Since the set of sizes
$$ \left\{ n \geq 0 \ | \ 1 + \frac{(j-2)m-j}{2j}\, n \leq g \right\} $$ is finite and since there is a finite number of automata with prescribed size $n$ and prescribed alphabet size $m$, the claimed result follows.

\subsection{Proof of Theorem \ref{th:genus-and-emulator} (Equivalence with genus $g$ directed emulators)}

It suffices to prove the equivalence between $(1)$ and $(3)$. Suppose that $g(L) \leq g$. Then there exists $\A \in {\rm{DFA}}$ 
such that $L(\A) = L$ and $g(\A) \geq g$. By Prop. \ref{prop:functor-automata-emulator}, ${\mathscr{G}}(\A)$ is a directed emulator over ${{\mathscr{G}}}(\Amin)$. By definition (\S \ref{sec:automata_and_graphs}),
$g({\mathscr{G}}(\A)) = g(\A) \geq g$. Therefore, $G(L) = {{\mathscr{G}}}(\Amin(L))$ has a directed emulator of genus $ \leq g$.

Conversely, let $\pi:G' \to G(L)$ be a directed emulator where $G'$ is a simple digraph of genus $ \leq g$. By Lemma \ref{lem:full_from_automata_to_directed_emulator}, the
directed emulator map lifts to an automaton morphism $\A \to \Amin(L)$, with $g(\A) = g(G') \leq g$. Hence $g(L) \leq g(\A) \leq g$.

\subsection{Cycles and directed emulators} \label{sec:cycles}


\begin{lemma} \label{lem:nosimplecycle}
 Let $k \geq 1$. Assume that a directed graph $G$ has no simple cycle of length less than or equal to $k$. Then neither has any directed emulator $\tilde{G}$ over $G$.
 \end{lemma}

\begin{proof}
Suppose that $\tilde{G}$ has a simple cycle $c'$ of length $l \leq k$. Its image in $G$ is a closed path $c$ of length $l$. The closed path $c$ admits a decomposition into a product of
simple cycles, each of which has length less than or equal to $l \leq k$.
\end{proof}

\begin{lemma}
The property for a deterministic automaton $\A$ to have no simple cycle of length $l \leq k$ is a property of the language $L(\A)$.
\end{lemma}

\begin{proof}
Let $\Amin$ be the minimal automaton of $L(\A)$.
Set $G = {\widetilde{\mathscr{G}}}(\Amin)$: the digraph $\tilde{G} = {\widetilde{\mathscr{G}}}(\A)$ is a directed emulator of $G$ (by Prop. \ref{prop:functor-automata-emulator}).
\end{proof}

\begin{lemma} \label{lem:preimage-cycle}
Let $\pi:\tilde{G} \to G$ be a directed emulator map. Let $c$ be a simple cycle of length $k$ in $G$. Then $\pi^{-1}(c)$ contains (a simple path followed by) a simple cycle of length a multiple of $k$ in $\tilde{G}$. Furthermore, this cycle respects the direction of $\tilde{G}$ if $c$ respects the direction of $G$.
\end{lemma}

\begin{proof}
With loss of generality, we can work in the category of simple digraphs. Let $c = v_{1} \cdots v_{n}$ where the $v_{i}$'s denote the vertices. Choose an arbitrary lift $v_{1}'$ of $v_{1}$. Since $\pi$ is a directed emulator map, each edge $e_{i} = v_{i}v_{i+1}$ of $c$ has a lift starting at any lift $v_{i}'$ of $v_{i}$. Lifting each edge of $c$ in this fashion, we obtain a path $c'=v_{1}' \cdots v_{n}'v_{1}''$ whose initial and final vertex lie in the same fibre: $\pi(v_{1}') = \pi(v_{1}'') = v_{1}$. If the initial and final vertices of $c'$ coincide, we stop and $c'$ is a simple cycle. Otherwise, starting again with $v_{1}''$ as a lift of $v_{1}$, we continue the process of lifting edges until we reach a first vertex $w$ that has already been reached. This implies that there is a path in $\tilde{G}$, which lifts $c$, starts and ends at $w$. Among the closed paths lifting $c$, let $\tilde{c}$ be closed path of minimal length with this property. Suppose that $\tilde{c}$ is not simple. Then there is an edge $e'$ of $\tilde{c}$ which is repeated, which implies that there are two pairs of vertices that are repeated, contradicting that $\tilde{c}$ has minimal length. Thus $\tilde{c}$ is simple. Since moreover $\tilde{c}$ covers $c$, the length of $\tilde{c}$ is a multiple of that of $c$.
\end{proof}


\subsection{Toric languages on a two-letter alphabet} \label{sec:proof-prop-toric-two-letters}

This section is devoted to the proof of Prop. \ref{prop:toric-two-letters}. Since the minimal automaton $\A_{\rm{min}}$ of $L_{n,p}$ embeds in a torus, $g(L_{n,p}) \leq 1$. We have to prove that $L_{n,p}$ is nonplanar, i.e. $g(L_{n,p}) \geq 1$, for $n, p \geq 4$. Let $\A$ be a finite deterministic automaton for $L_{n,p}$. By Prop. \ref{prop:functor-automata-emulator}, $\A$ over $\A_{\rm{min}}$ induces a directed emulator map $\pi: \widetilde{\mathscr{G}}(\A) \to \widetilde{\mathscr{G}}(\A_{\rm{min}})$. Since $n,p \geq 4$, $\widetilde{\mathscr{G}}(\A_{\rm{min}})$ has no simple cycle of length $\leq 3$. Applying Lemma \ref{lem:nosimplecycle}, we see that neither has ${\widetilde{\mathscr{G}}}(\A)$. Consider now a minimal embedding of $\A$ into some closed oriented surface $\Sigma$. 

\noindent {\emph{Claim.}} Each face of the embedding has length at least $4$. 

\noindent {\emph{Proof of the claim}}. Suppose the contrary. There is a face $f$ of length $\leq 3$. The boundary of $f$ must be a nonsimple cycle $c = \partial f$. Since there is no self-loop in $\A$, $c$ has length $2$ and consists in exactly one edge $e$ with immediate backtracking. It follows that $e$ is monofacial. Let $\vec{e}$ be the original oriented edge in $\A$. Since $e$ is monofacial, one of the endpoints, $s(\vec{e})$ or $t(\vec{e})$ has total degree $1$, which is a contradiction.

For $j \geq 1$, let $f_{j}$ be the number of faces of length $j$.
By the claim above, $f_1 = f_2 = f_3 = 0$. Therefore, according to the genus formula \cite[Th. 5]{BD},
$$ g(L) = g(\A) = 1 + \sum_{j \geq 1} \frac{j-4}{8}f_{j} = 1 + \sum_{j \geq 4} \frac{j-4}{8} f_{j} \geq 1.$$
This is the desired result.

\begin{remark}
Inspection of the proof shows that if $L$ is a two-letter language
that does not have any simple cycle of length $\leq 3$, then $g(L) \geq 1$.
\end{remark}

\vspace{0.2cm}

\noindent {\textbf{Acknowlegments}}. This work was partially supported by the CIMI (Centre International de Math\'ematiques et d'Informatique), ANR-11-LABX-0040-CIMI within the program ANR-11-IDEX-0002-02 and by ANR-14-CE25-0005-3, ELICA.

\bibliographystyle{alpha}
\bibliography{NGLR}

\vspace{0.2cm}

\noindent G. Bonfante, LORIA, Universit\'e de Lorraine, Campus scientifique, B.P.~239, 54506 Vandoeuvre-l\`es-Nancy Cedex

\vspace{0.2cm}

\noindent {\tt{guillaume.bonfante@loria.fr}}\\

\vspace{0.2cm}
\noindent F. Deloup, IMT, Universit\'e Paul Sabatier -- Toulouse III, 118 Route de Narbonne, 31069 Toulouse cedex 9

\vspace{0.2cm}

\noindent {\tt{florian.deloup@math.univ-toulouse.fr}}\\
%

\end{document}